\documentclass{article}
\usepackage{arxiv}
\usepackage{authblk} 
\usepackage[utf8]{inputenc} 
\usepackage[T1]{fontenc}    
\usepackage{hyperref}       
\usepackage{url}            
\usepackage{booktabs}       
\usepackage{amsfonts}       
\usepackage{nicefrac}       
\usepackage{microtype}      
\usepackage{lipsum}
\usepackage{verbatim}
\usepackage{mathrsfs}
\usepackage{amsthm}
\usepackage{blindtext}
\usepackage{tikz}
\usetikzlibrary{decorations.text,calc,arrows.meta}
\usetikzlibrary{backgrounds}
\usepackage{pifont}
\usepackage{fancyhdr}       
\usepackage{graphicx}       
\graphicspath{{media/}}     
\usepackage{amsmath, amssymb}
\newtheorem{theorem}{Theorem}

\theoremstyle{definition}

\usepackage{cite}
\usepackage[boxed]{algorithm2e}
\SetAlCapSkip{1em}
\usepackage{graphicx}
\usepackage{textcomp}
\usepackage{xcolor}
\usepackage{amsthm}
\def\BibTeX{{\rm B\kern-.05em{\sc i\kern-.025em b}\kern-.08em
    T\kern-.1667em\lower.7ex\hbox{E}\kern-.125emX}}
\newcommand{\TSPL}{\mbox{TSPL}}

\newcommand{\cut}[1]{{}}
\title{Croesus: Multi-Stage Processing and Transactions for Video-Analytics in Edge-Cloud Systems
}

\author[1]{Samaa Gazzaz}
\author[2]{Vishal Chakraborty}
\author[2]{Faisal Nawab}

\affil[1]{\footnotesize University of California, Santa Cruz}
\affil[2]{\footnotesize University of California, Irvine}

\begin{document}
\maketitle

\begin{abstract}
Emerging edge applications require both a fast response latency and complex processing. This is infeasible without expensive hardware that can process complex operations---such as object detection---within a short time. Many approach this problem by addressing the complexity of the models---via model compression, pruning and quantization---or compressing the input. In this paper, we propose a different perspective when addressing the performance challenges. Croesus is a multi-stage approach to edge-cloud systems that provides the ability to find the balance between accuracy and performance.
Croesus consists of two stages (\textcolor{black}{that can be generalized to multiple stages}): an initial and a final stage. The initial stage performs the computation in real-time using approximate/best-effort computation at the edge. The final stage performs the full computation at the cloud, and uses the results to correct any errors made at the initial stage. In this paper, we demonstrate the implications of such an approach on a video analytics use-case and show how multi-stage processing yields a better balance between accuracy and performance. Moreover, we study the safety of multi-stage transactions via two proposals: multi-stage serializability (MS-SR) and multi-stage invariant confluence with Apologies (MS-IA).
\end{abstract}


\section{Introduction}
Modern object detection models
are based on complex Convolutional Neural Networks~(CNN) that require
GPU clusters costing tens of thousands of dollars to perform object
detection in real-time~\cite{noscope,wu2019fbnet,he2018amc,cai2019once}. 
This is infeasible for edge applications that require real-time processing
but cannot afford to place expensive hardware at the edge.
Furthermore, many of these applications require response in the
scale of milliseconds (such as V/AR~\cite{lincoln2016motion} and smart city
Vehicle-to-Everything~\cite{chen2017vehicle}). This prohibits the use of faraway cloud resources.

There is a large body of research in the machine learning community that aims at addressing the trade-off between accuracy and performance in deep learning (DL) models by utilizing compression, pruning and quantization techniques ~\cite{wu2019fbnet,he2018amc,han2019deep,cai2019once,kim2019efficient,luo2017thinet,ullrich2017soft,chen2018constraint,xu2018deep,Choi2020universal,dubey2018coreset}. 
%
%
In these approaches, we notice a trade-off between accuracy and performance. The accuracy of a compressed model is typically lower compared to the full model while performance is improved dramatically. For example, in~\cite{wu2019fbnet}, the compressed model improves latency from 23.1 ms to 2.9 ms, while lowering the accuracy from $74.1\%$ to $50.2\%$. 
%
%
Other papers in the field of image compression aid in reducing the amount of time needed to process data~\cite{8305033,Liusanfran,8476610}. 
other researchers opt to specializing DL models for certain use cases to improve performance~\cite{julian2019deep,8354254,lawhern2018eegnet,guo2021compact}. 

An important aspect that is overlooked in many video analytics solutions is that they are not integrated with the system's data processing and management. Video analytics generates insights from videos that would typically be used in a data management application. For example, detecting objects in V/AR might feed into a mobile game, immersive social network, or other application. 
We propose Croesus, a multi-stage edge-cloud video processing framework that aims to manage the performance-accuracy trade-off in DL models. \color{black}The framework consists of an edge-cloud video analytics component and a transaction processing component. Each component may exist in isolation of the other and benefit other use cases, however, they are co-designed to achieve the goals of data management for video analytics applications. \color{black}
This proposal separates computation into two stages: an initial stage that depends on best-effort computations at the edge (using a fast but less accurate DL model), and a final stage at the cloud to correct any errors incurred in the initial stage (using the accurate but slower DL model.) 
For example, for object detection in applications such as V/AR, instead of depending solely on the full CNN model, a more compact model is used at the edge to respond immediately to users. If needed, some frames are sent to the full CNN model on the cloud to detect any errors on the immediate responses sent by the initial stage. If an error is detected, then a correction process is performed in the final stage. The mechanism to correct errors is an application-specific task and our method allows flexibility in how errors are corrected.
The advantage of this model is that users have the illusion of both a fast and accurate object detection. The downside is the possibility of short-term errors. 
This pattern of the multi-stage model is useful for applications that require fast response but where the full model cannot be used within the desired time restrictions.

{\color{black}We formalize and analyze the transactions (a transaction is a group of database read/write operations that represents a task or a program) in Croesus using a formal \emph{multi-stage transaction} model.} Our model divides transactions into two sections: an initial and a final sections {\textcolor{black}{(we also show how this model can be extended to multiple sections)}}. The initial section is responsible for updating the system using the results of the initial object detection stage, and the final transaction is responsible for finalizing/correcting state using the results of the final (object detection) stage. \color{black}The multi-stage transaction model can be generalised to have more than two stages. However, our analysis with the general design turned out to add additional overhead without providing a significant benefit for edge-cloud video analytics. The reason is that the asymmetry in edge-cloud systems is two-fold: in the edge (low-capability, real-time requirement) and in the cloud (high-capability, less stringent latency requirement).\color{black} 

The multi-stage transaction model leads to challenges when reasoning about the correctness guarantees that should be provided to users. This is because the multi-stage transaction model breaks a fundamental assumption in existing transaction models, which is the assumption that a transaction is a single program or block of code. Therefore, there are challenges on coming up with an abstraction of initial and final sections and how they interact. Also, there is a need to specify what makes an execution of initial and final sections correct in the presence of concurrent transactions. We cannot reuse existing correctness criteria---such as serializability~\cite{bernstein1987concurrency}---as they would not apply to the multi-stage transaction model.

For those reasons, we propose a multi-stage transaction processing
protocol and study the safety-performance trade-offs in multi-stage transactions. We investigate two safety guarantees:
\emph{(1)~Multi-stage Serializability (MS-SR)}, which mimics the
safety principles of serializability~\cite{bernstein1987concurrency} by requiring that each transaction would be isolated from all other transactions. 
\emph{(2)~Multi-stage Invariant
Confluence with Apologies (MS-IA)}, which adapts invariant confluence~\cite{bailis2014coordination} and apologies~\cite{DBLP:conf/cidr/HellandC09} to the multi-stage transaction model and enjoys better performance characteristics and flexibility compared to MS-SR.
The multi-stage transaction pattern of Croesus invites a natural method of adapting invariant confluence and apologies. In particular, the final section is---by design---intended to fix any errors caused by the initial stage. This can be viewed as the final stage ``correcting any invariant violations'' and issuing ``apologies'' for any erroneous work generated by the initial section.

In the rest of this paper, we present background in
Section~\ref{sec:background}, followed by the design of Croesus (Section~\ref{sec:video}) and multi-stage transactions (Section~\ref{sec:safety}). 
Experiments and related work are presented in Sections~\ref{sec:eval}
and~\ref{sec:related}, respectively. 
The paper concludes in
Section~\ref{sec:conclusion}.

\begin{figure}
\centering
\includegraphics[scale=0.40]{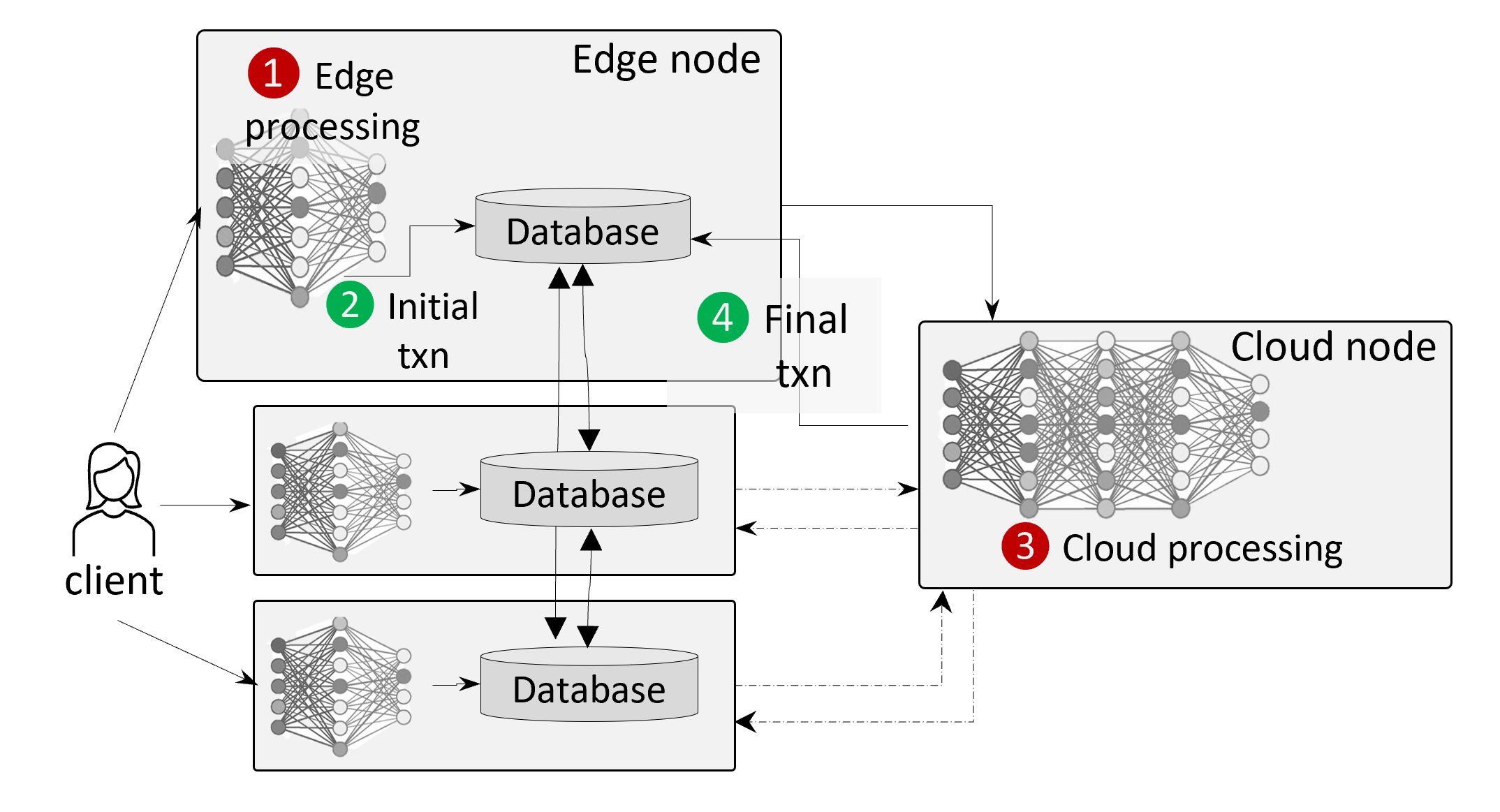}
\vspace{-0.2in}
\caption{\textcolor{black}{Croesus' execution pattern}}
\label{fig:exeptrn}
\vspace{-0.15in}
\end{figure}

\section{Background}
\label{sec:background}

In this section, we present background on the multi-stage system model and object detection.

\subsection{System and Programming Model}
\label{sub:model}

\textbf{Edge-Cloud Model.}
{\color{black}Our proposed system model consists of edge nodes and a cloud node (see Figure~\ref{fig:exeptrn}). 
Each edge node maintains the state of a partition (database transactions are performed on the partition copy at the edge.) For ease of exposition, we focus on a single edge node and partition in this paper. In edge applications, interactions between users tend to have spatial locality and are therefore typically homed in the same edge node and partition.}

\textbf{Application Model.}
The applications we consider are video-driven---meaning that the input and triggers to operations on data are done via a video interface. For example, a gesture or object detected on a V/AR headset triggers a database transaction. This translates to the following processing steps for each frame $f$: (1)~the frame $f$ is processed using the small model on the edge node, $M_e$, to generate labels (labels are the detected objects and/or actions). We call these the edge labels and are denoted by a set $L_e$. (2)~the edge labels $L_e$ are used to trigger transactions that take the labels as input. These transactions are denoted by the set $T_f$. The initial sections of each of these transactions in $T_f$ are processed to return an immediate response to users and potentially write to the database on the edge node. (3) concurrently, the frame $f$ is also processed in the original, more accurate object detection model on the cloud, denoted by $M_c$. Once the cloud model generates the labels, denoted by $L_c$, they are sent to the edge node. (4)~when the labels $L_c$ from the cloud are received, they are used to trigger two types of events. The first is to trigger the final sections of the transactions $T_f$ that started for frame $f$. The input to these sections is the correct label(s) of the object(s) that triggered the transaction. The second is to trigger new transactions that should have been triggered by the frame but their labels where missing in $L_e$. 
We focus on the first pattern as the second pattern can be viewed as a subset of the first. 

\textbf{Example Application.}
Consider a smart campus Augmented Reality (AR)
application with two basic functionalities: (1)~Task 1: continuously, an
object detection CNN model detects buildings in the campus. If a
building is detected, the database is queried and information
about the building---such as available study rooms---is augmented
onto the headset view. (2)~Task 2: if the user clicks on
an auxiliary device, a study room is reserved in the currently
detected building. 

\textbf{Execution Pattern.}
The execution pattern of this application is the following (shown in Figure~\ref{fig:exeptrn}):
The headset captures images continuously and sends them
to the nearby edge node. The edge node performs the initial stage of
computation by running the captured frame, $f$ on the small (fast but
inaccurate) DL model, $M_e$ (step 1). The labels extracted from the model, $L_e$, are used
to trigger the initial section of transaction $T_f$ (step 2). For example, if the engineering building is
detected, then the transaction's initial section reads information about the building.
The outcome of this transaction is sent back to the headset to be
rendered and augmented onto the display. During this time, the frame
is forwarded to the cloud node which runs the full (slow but
accurate) CNN model, $M_c$ (step 3). The labels, $L_c$ extracted from the model are sent
back to the edge node. Once the edge node receives the correct
labels, it performs the final stage of the transactions in $T_f$ (step 4). The final
stage takes as input both the original detected labels in the initial
stage as well as the new, correct, labels.  

\textbf{Programming Interface.}
The programming model exposes an interface to write both the initial
and final sections of the transaction. In our application for example,
there are two transactions, one for each task.
For task 1 (display
information about detected buildings), the initial section is
triggered for each frame with a label in the class ``building'' and it takes as
input the detected labels, $L_e$.
For each detected label, the initial section
reads the information about that key from the database and returns it
to the headset to be rendered.
The final section is triggered after the correct labels, $L_c$, are sent from
the cloud node. It checks if the labels are the same; if they are,
the transaction terminates (note that the decision to terminate is specific to this example transaction, but other application might use the final section to perform some final actions even if the labels were correctly detected in the initial stage.) If they are not, then the transaction
reads the labels of the correct detected building and sends them to the
headset to render the correct information and an apology.
\footnote{In a real application, the corrected information would also influence the small model---via retraining and heuristics such as smoothing---so that the error would not be incurred in the following frames.}.

For task 2 (reserve a study room), the initial section is triggered
when the auxiliary device is clicked by the user. The initial
section takes as input the most recent detected labels and their
coordinates. If there are more than one label, the initial section
picks the label that is closest to the center of the frame. Then, the
initial section reserves a study room if one exists.
The final section---triggered after receiving the correct
labels---checks if the center-most label matches the building where
the study room was reserved. If so, the transaction terminates.
Otherwise, the original reservation is removed from the database
and---if available---a new reservation with the right building is
made. The results are sent back to the AR headset to be rendered with
an apology. 

\subsection{Accuracy-Performance Trade-off in Object Detection}

\textbf{Convolutional Neural Networks (CNNs).}
A CNN is designed and trained to detect labels of objects in an input frame. Different CNN models have different structures and variations, and we refer the interested reader to these surveys~\cite{khan2020survey, sindagi2018survey}.  Our work applies to a wide-range of CNN models as we use them as a black box. 

\textbf{Accuracy-Performance Trade-off.}
The complex processing of CNNs result in higher inference time. 
It is estimated that running a state-of-the-art CNN model in real-time requires a cluster of GPUs that costs tens of thousands of dollars~\cite{noscope}. This means that running a CNN model on commodity hardware---such as what is used in edge devices---would lead to prohibitively high latency. 
This led to exploring the accuracy-performance trade-off in CNN models. Specifically, there has been efforts to produce smaller CNN models that would run faster on commodity hardware~\cite{yolo9000,lawhern2018eegnet,wen2019memristor,zhang2020compact,racki2018compact,xu2019accurate}. The downside of these solutions is that they are less accurate than full CNN models. In this work, we aim to utilize both small and full CNN models by using small models for fast inference and original models to correct any errors.

\textbf{Derivative Models.}
The interest in the accuracy-performance trade-off in CNNs led to efforts that enable deriving smaller---faster---models using existing original CNN models. One approach is to use a smaller model that handles the same scope of labels of the original model but with less accuracy~\cite{yolo9000}. Another approach is to create smaller---specialized---models that narrow the scope of labels to enable faster inference while retaining accuracy for the select labels~\cite{noscope}.
In our work, we consider both variations. For smaller, less accurate models, the Croesus pipeline helps correct errors due to inaccuracy and for specialized models, the Croesus pipeline helps correct errors due to the narrower scope of labels.

\section{Croesus Design}
\label{sec:video}

In this section, we present the design of Croesus and an optimization that controls the accuracy-performance trade-off.

\subsection{Overview}

\textbf{System Model.}
The system model of Croesus (Section~\ref{sec:background}) consists of an edge node and a cloud node.
The edge node hosts a small CNN model denoted by $M_e$ that is used to perform initial processing. The edge node also hosts the main copy of it's partition's data. The edge node processes both the initial section and the final section. The initial section of a transaction is triggered by the labels of the model on the edge, $M_e$, and the final section is triggered by the labels of the model on the cloud, $M_c$.
The execution pattern of requests is shown in Figure~\ref{fig:exeptrn} and described in Section~\ref{sub:model}.

\textbf{Workflow.}
The workflow of requests in Croesus is the following: a frame $f$ is sent from the client to the edge node. The edge node processes $f$ using the edge model, $M_e$. The labels from $M_e$, $L_e$, are used to trigger corresponding transactions, $T_f$ (the programmer defines what transactions should be triggered for each class of labels.) The initial sections of transactions in $T_f$ are processed on the edge node. At this time, the response from the initial sections are sent to the client. This marks the initial commit stage. In the meantime, the frame $f$ is sent to the cloud node. Once the cloud node receives it, the cloud model, $M_c$, is used to process $f$. The corresponding labels, $L_c$, are then sent to the edge node. When the edge node receives the cloud labels $L_c$, the final sections of transactions in $T_f$ are triggered. The responses and apologies from these final sections are sent to the client. This marks the final commit stage.

\textbf{Bandwidth Thresholding.}
The pattern of edge-cloud stages introduces a bandwidth overhead due to the need to send all frames from the edge to the cloud. This can be problematic due to the high overhead on the edge device and the monetary cost of communicating data to the cloud.
(e.g., some public cloud providers charge a cost for communicated data between the data center and the Internet).
To this end, we tackle the problem of limiting edge-to-cloud communication. We use the confidence of the labels that are generated by the edge model, $M_e$, to decide whether we need to send the frame to the cloud or not. 
Specifically, if the edge model's confidence is high enough, this is an indication that the detected labels are more reliable than other detections that have less corresponding confidence. Later in this section, we develop a bandwidth thresholding mechanism to investigate sending frames to the cloud selectively using the edge model's confidence.

\subsection{Initial-Final Section Interaction}
\label{sub:init-final}
A unique property of multi-stage processing is that there are two stages of processing where the first stage is fast and less accurate and the second is slow and accurate. This property leads to the need to understand how they interact and what guarantees should be associated with each stage. In the rest of this section, we provide such properties that are useful to programmers in the multi-stage model.
In the initial stage, the initial section of a transaction, $s_i$, uses the input from the edge model, $M_e$, to generate a response to the user. This response represents an \emph{initial-stage commit}. The initial-stage commit---when received by a client---represents the following: (1)~the response is a preliminary and/or best-effort result of the transaction. (2)~any errors in this initial processing will be corrected by the logic specified by the programmer in the corresponding final section. This second property is critical because it leads to having to enforce a guarantee that if the initial section of a transaction returns a response to the client (an initial-stage commit), then the underlying system must guarantee that the corresponding final section would commit as well. 
This is trivial for a transaction running by itself, however, when transactions are running concurrently, this leads to complications. (In Section~\ref{sec:safety}, we present the concurrency control mechanisms for multi-stage transactions where we encounter these complications.) 

When the final section of the transaction starts, it is anticipated for the final section to observe what the input labels were to the initial section---to know whether the input was erroneous---and what the initial section did---to know what to fix if an error was detected. To avoid adding complexity to the system model and description, we consider that these two tasks are performed by the programmer using database reads and writes. Specifically, the initial section communicates to the final section via writing its input and state to the database.

\subsection{Algorithms}

Now, we provide the detailed algorithms of Croesus. 
Parts of the algorithms use a concurrency control component that we present and design in Section~\ref{sec:safety}. We will denote this concurrency control component as \textsf{CC} and a transaction block would either be \textsf{CC.initial\{ \}} for an initial section and \textsf{CC.final\{ \}} for a final section. Both transaction blocks get the detected labels as input, but we omit it for brevity. 

\subsubsection{Client Interface}
The client captures frames, gets user input (from auxiliary devices), and displays responses. 
For example, in a V/AR application, the client captures a frame from the headset camera and sends it to the edge node. Likewise, if there are any associated auxiliary or wearable devices, the client sends the input/commands that correspond to these devices. This process of sending frames and input is continuous---there is no blocking to get the response from the edge node. When a response is received from the edge node, that response is rendered and augmented in the user's view. 

\subsubsection{Edge Node Algorithms}
The edge node is responsible for the initial stage of processing (using the small model $M_e$), transaction processing, and storage. 
There are two main components in the edge node: the input processing component and the transaction processing component. 
The following is a description of the main tasks that are handled by the edge node.

\textbf{Initialization and Setup.}
Starting an edge node includes setting up a small model, $M_e$, a data store $ds$, and a \emph{transactions bank}. The small model $M_e$ is the one that will be used to process incoming frames. The transaction bank is a data structure that maintains the application transactions and what triggers each transaction. For example, an application may have a transaction $t_{bldng}$ that reads the information about a building that is detected in a frame. The transaction $t_{bldng}$ takes as input the label that is associated with a building. The transactions bank helps the edge node know which transactions should be triggered in response to a label. 
For example, if a label $l_1$ represents a label name ``Engineering Building'' and label $l_2$ represents a label name ``University Shuttle 42'', the transaction $t_{bldng}$ should be triggered in response to $l_1$ but not $l_2$.

The way the transactions bank helps in making this decision is that it maintains a table, where each row corresponds to a class of labels and the transactions that would be triggered from that class of labels. For example, a row in that table can have a class of labels called ``Buildings'' and it contains all the labels that would correspond to a building. That row would also have $t_{bldng}$ and any other transactions that should be triggered in response to the ``Building'' class. A row in the transactions bank may also have other associated triggers, For example, a transaction $t_{rsrv}$ that is used to reserve a study room in a building would be triggered if both a building label is detected in the frame \emph{and} the auxiliary device input is received.

\textbf{Input and Initial Stage Processing.}
The initial stage processing represents the input processing using the small model, $M_e$, in response to a received frame or user input. When a frame $f$ is received by the edge node, it is supplied to the small model $M_e$. The model $M_e$ returns a set of labels $L^f_e$. Each label, $L^f_e[i]$, consists of the the name of the label, $L^f_e[i].name$, the confidence of the label, $L^f_e[i].confidence$, and the coordinates of the label, $L^f_e[i].coordinates$. The input processing component removes any labels from the set $L^f_e$ that have low confidence (the threshold for a low confidence is a configuration parameter.)  Finally, the input processing component gets the information of all the transactions that correspond to the detected labels, $L^f_e$, by reading from the transactions bank. The set of triggered transaction, $t_f$, is sent to the transaction processing component.

Similar to how frames trigger transaction, when a different input is received by the input processing component---such as a click on the auxiliary device---the input processing component generates the set of transactions $t_e$ that corresponds to the input. An auxiliary input might lead to an action that is independent from the captured frame. For example, a click on the menu button may display the menu and general user information. In this case, the entry in the transactions bank is only specified by the input type. Alternatively, the input might be coupled with a specific label class to trigger a transaction. For example, a click would display a captured building's information using $t_{rsrv}$. In such a case, $t_{rsrv}$ would only be triggered if both the click and a building label are detected. To facilitate such actions, the input processing component matches a received auxiliary input with the labels from the most recently detected labels. 

After transactions, $t_f$, are sent to the transaction processing component (TPC), the frame $f$ is sent to the cloud node to be processed using the cloud model, $M_c$. This concludes the tasks performed for input processing.

\textbf{Initial Transaction Section.}
When the input processing component generates the set $t_f$ for a frame $f$, these transactions are sent to the TPC. The TPC then triggers the initial section of these transactions. The read and write operations to the database are managed by the concurrency control component by wrapping them in the \textsf{CC.initial\{ \}} block. (The implementation details of the concurrency control component are presented in Section~\ref{sec:safety}). The initial section of a transaction $t$ would either commit or abort---based on the decision of the concurrency controller. If the initial section aborts, then the abort decision is sent to the client. Otherwise, the response from the initial section is sent to the client, which represents the initial commit point for $t$. The TPC records the decision for the initial section with the labels, $L^f_e$, and waits until the corresponding labels are received from the cloud model. 

\textbf{Final Transaction Section.}
After processing the initial section, the TPC waits for the correct labels, $L^f_c$, from the cloud node. Once received, the following is performed for each label, $L^f_e[i]$ in $L^f_e$. The label $L^f_e[i]$ is matched with a label in $L^f_c$. The matching is performed by finding if the bounding box (represented by the x-y coordinates) of a label in $L^f_c$ overlaps with the bounding box of $L^f_e[i]$. The overlap does not need to be exact---if the label overlap in more than X\%, where X is a configuration parameter, then the two labels are considered overlapping. If there are more than two candidates in $L^f_c$ that overlap with $L^f_e[i]$, then the one with the bigger overlap is chosen. There are the following cases of matching the label $L^f_e[i]$ to a label in $L^f_c$: 
(1)~If an overlapping label cannot be found in $L^f_c$, then the label $L^f_e[i]$ is considered erroneous and the final section of the corresponding transaction is called with an empty label.
(2)~If there is a label in $L^f_c$ that overlaps with $L^f_e[i]$ and the label name is the same. In that case, the label $L^f_e[i]$ is considered correct and the final section of the corresponding transaction is called with the same label.
(3)~If there is a label in $L^f_c$ that overlaps with $L^f_e[i]$ and the label names are different. In that case, the label $L^f_e[i]$ is considered erroneous and the final section of the corresponding transaction is called with the overlapping label from $L^f_c$.

Once this matching process is complete, then the TPC checks if there are any labels in $L^f_c$ that were not matched. For each one of these labels $L^f_c[i]$, the TPC triggers an initial section and final section with the label in $L^f_c[i]$. 

\subsubsection{Cloud Node Algorithms}
The cloud node has a single task of processing frames using the cloud model, $M_c$. When a frame $f$ is received from an edge node, the labels, $L^f_c$, are derived using $M_c$ and then sent back to the edge node.

\subsection{Bandwidth Thresholding} 

A major problem faced by video-analytics applications in the edge-cloud paradigm is the high edge-cloud bandwidth consumption due to the large size of videos. Sending all frames from the edge to the cloud poses a performance challenge due to the communication overhead as well as a monetary overhead due to the cost of transferring data between the edge and the cloud (most public cloud providers charge applications for data communication between the cloud and the Internet). 
We extend our solution to reduce the reliance on cloud nodes with the goal of overcoming the performance overhead and monetary costs of edge-cloud communication.

The observation we utilize to reduce edge-cloud communication is that we can use the confidence of edge computation to decide whether verifying with the cloud node is necessary. (Confidence here represents the statistical confidence generated by CNN models which is a typical feature of such models.) Specifically, if the confidence of the produced detections in the edge model, $M_e$, is high, it is likely that the edge model produced correct labels. Therefore, it would not be necessary to send the frame to the cloud. Likewise, if the detections had extremely low confidence, then it is likely that these are erroneous, false detections, and thus sending the frame to the cloud node would be unnecessary as they can be discarded immediately. What is left are detections that have confidence values that are not too high and not too low. These detections are ones that likely indicate the presence of an object of interest, but its label might be incorrect.

More formally, we represent with $\theta_L$ and $\theta_U$ the lower and the upper confidence thresholds such that $0 \leq \theta_L < \theta_U < 1$. Generally, an object with confidence lower than $\theta_L$ is discarded as being likely a false-positive (this is called the \emph{discard interval}). An object with confidence higher than $\theta_U$ is assumed to be correct and is not sent to the cloud node (this is called the \emph{keep interval}). Objects with a confidence between $\theta_L$ and $\theta_U$ are sent to the cloud for validation (this is called the \emph{validate interval}).
However, there is a challenge in adopting this model as it is not clear how to derive these confidence thresholds to preserve the integrity of the underlying models. Specifically, a \emph{performance-accuracy trade-off} controls this decision. A large validate interval would lead to better accuracy, since more frames are sent to the cloud for validation and correction. Likewise, a small validate interval would lead to worse accuracy but better performance in terms of average latency and edge-cloud bandwidth utilization. This is complicated further because the size of the validate interval is not the only factor controlling this trade-off. The validate interval size may lead to different performance-accuracy trade-offs based on where it is located in the threshold space from 0--100\%.


%
\textbf{Optimization Formulation.}
The input to the optimization problem is a set of video frames $V = \{v_1, \hdots, v_n \}$, and an object query $O$ (e.g., bus), which needs to be detected in the frames.
Let $n_i$ be the number of instances of object $O$ detected in frame $v_i$ 
(by the NN in the edge-node) with confidence $\beta_i = (\beta_i^1, \hdots, \beta_i^{n_i})$ where $\beta_i^k$ is the confidence corresponding to the
$k^{\text{th}}$ instance of object $O$, for $1 \leq k \leq n_i.$ We denote this as \emph{edge-confidence}. 

Let $m =~|\{v_i \in V \mid \exists k \text{ s.t. } \theta_L \leq \beta^k_i \leq \theta_U\} |$ be the number of frames which were sent to the cloud. We define the ratio $\delta(\theta_L, \theta_U) = \frac{m}{n}$ (where $n$ is the number of frames in $V$) and have the corresponding F-score $f({\theta_L, \theta_U}) = \frac{2pr}{p+r}$ where $p$ is precision and $r$ is recall. We want to find $(\theta_L, \theta_U)$ such that $\delta(\theta_L, \theta_U)$ is minimized and the corresponding $f ({\theta_L, \theta_U}) \geq \mu.$ Let $\mathbb{S} = \{x \in \mathbb{R} \mid  0 \leq x < 1 \}$. We have:
\begin{align}
    T = \underset{(x,y) \in \mathbb{S}^2, \mu}{\text{argthresh }} f(x,y) := \{ (x, y) \in \mathbb{S}^2 \mid f({x, y) \geq \mu}\} \label{eqn:T}
\end{align}
\begin{align}
    (\theta_L, \theta_U) &= \underset{(x,y) \in T}{\text{argmin }} \delta (x,y)  \nonumber
    \\&:=
    \{(x^*,y^*) \in T \mid \forall (x,y) \in \mathbb{S}^2, \delta(x^*, y^*) \leq \delta(x,y) \}. \label{eqn:Theta}
\end{align}
This formulation produces the thresholds $(\theta_L, \theta_U)$ given $\mu$.

{\color{black}
\subsection{Generalizing Multi-Stage Processing}

In this section, we have focused on models with two stages. This is because the application domain we consider has a two-tier symmetry that invites the use of two sections, one that represents the edge and another that represents the cloud. However, the multi-stage processing model can be utilized for other use cases where the asymmetry has more than two levels. Our designs and treatments can be extended to these cases as we describe in the rest of this section.

\textbf{Model.}
In a general multi-stage model, there are $m$ stages, $s_0,\ldots,s_{m-1}$. The first stage, $s_0$, represents the initial stage of processing and the last stage, $s_{m-1}$, represents the final stage of processing. All other stages are intermediate stages. The data storage is maintained by the node handling stage $s_0$. Each stage contains a video/image detection model---where typically the model at stage $s_i$ (denoted $m_i$) has better detection that model $m_j$, where $j<i$. A transaction consists of $m$ sections, each one ($t_i$) corresponding to a stage ($s_i$).

\textbf{Processing.}
When a frame $f$ is received, it is first sent to the initial stage, $s_0$. The initial stage processes $f$ using $m_0$ and takes the outcome of the model to process the first section of the transaction $t_0$. Then, the frame is processed at the next stage $s_1$---using $m_1$---and the outcome is used to trigger transaction $t_1$. This continues until the final stage.
If bandwidth thresholding is performed at any stage, then the sequence from initial to final stages might be broken. For example, if at stage $s_i$, the bandwidth thresholding algorithm (as presented earlier in the section) decides that the frame does not need to be forwarded to the next stage, then the sequence stops and the remaining transaction sections are performed.

}

\section{Multi-Stage Transactions}
\label{sec:safety}


\subsection{Multi-Stage Transaction Model}
We consider a new multi-stage transaction model where every transaction comprises of two distinct sections: the initial section
and the final section. Each section, $s$---in a transaction $t$---consists of
read ($r_t^s(x)$) and write ($w_t^s(x)$) operations in addition to
control operations to begin ($b_t^s$) and commit ($c_t^s$) each
section. For example, consider a multi-stage transaction $t$. The execution of the transaction would look like the
following: $b_t^i\ r_t^f(x)\ w_t^i(y)\ c_t^i\ b_t^f\ w_t^f(z)\  c_t^f$
where $i$ stands for the initial section and $f$ stands for the final section. 

If the initial section of a transaction commits (called initial commit), then the final section must begin and commit (called final commit) as well. When we say that a transaction $t$ in our model has committed, we mean that both sections of $t$ have committed. Furthermore, the final section of a transaction cannot begin before the initial section. The case for conflicts of transactions also demands special consideration. In our model, we say two transactions to be conflicting if there is at least one conflicting operation in either of the sections. The seemingly simple abstraction of splitting every transaction into two sections complicates the basic notions of the general transaction model. In the following, we take a look at safety and describe two notions of consistency in our model. 

\subsection{Safety}
In the
absence of concurrent activity, safety is straight-forward; the
initial section is followed by the final section and both are
processed as the programmer expects. When concurrency---which is
important for performance---is introduced, it challenges the
programmer's notion of the sequentiality of running transactions and
multi-stage sections (other conflicting transactions may run within
and between a transaction's sections.)

\label{ex:multitxn_safety}
For example, consider an application where there are two transactions, $t_1$ and $t_2$, each of which increment the value of a data object $x$ by one. Suppose that, for each transaction, the initial stage consists of reading the value of $x$; the value is increased, and the new value is written in the final section. Therefore, if the two transactions executed concurrently and both $t_1$ and $t_2$ read the same value of $x$, then the final value of $x$ would only increase by one. This is an anomaly because there were two transactions that incremented the value of $x$ and the value of $x$ should have increased by two.

safety is different because it is also actions between sections not
only within a transaction.
safety here is also different than typical concurrency -- it is not
about conflicting copies to be merged, it is about a wrong trigger
or wrong input.
Evidently, multi-stage consistency adds to the complexity involved in
traditional consistency guarantees such as serializability in two
ways: (1)~multi-stage transactions consists of two separate stages.
This means that in addition to the concern of concurrent transactions
interleaving operations within each section, there is a need to
consider whether sections of transactions running \emph{between} the sections of other transactions should be permitted. 
(2)~in multi-stage transactions, inconsistency is not only due to
concurrent activity, but also due to erroneous
transactions that have an incorrect trigger or input (\emph{e.g.}, an
erroneously detected building in the edge stage of processing leads
to triggering the wrong transaction and/or supplying it with the
wrong input.) 

Due to these differences, we revisit transactional consistency in light of multi-stage transactions. We present and discuss two variants of multi-stage transaction consistency.
In both variants, we assume that traditional concurrency control mechanisms are used to ensure that each section is serializable relative to other transactions' sections. (This means that each section is atomic and isolated from other sections and that there is a total order on sections.) This leaves the novel challenge to safety that is introduced in our work, which is how these sections can be reordered relative to each other.

\vspace{-0.05in}

\subsection{Multi-Stage Serializability (MS-SR)}
In MS-SR, we mimic the safety principles of serializability, which is---informally---a guarantee that all transactions execute with the illusion of some serial order of all transactions~\cite{bernstein1987concurrency}. When trying to project this to multi-stage transactions, this translates to the requirement that all transactions are processed serially, where the final section of a transaction appears immediately after the initial section. This guarantee can be reduced to serializability by considering that the initial and final sections are part of the same serializable transaction. The main difference is that when the initial section commits, it is a guarantee that the final section would eventually commit---it cannot abort due to unresolved conflicts. As we will see in the rest of this section, this requirements complicates the processing of the initial section.

In order to specify MS-SR formally, we introduce some notations and state our assumptions.  We denote with $<_h$, the ordering relation on execution history of transaction sections. This relation represents the ordering relative to the commitment rather than the beginning of the section. For example, $s_a <_h s_b,$ denotes that the left-hand side is ordered before the right-hand side, i.e., section $s_a$ is ordered before section $s_b$. 

Consider two conflicting transactions $t_k$ and $t_j$ (i.e., they have at least one conflicting operation in either section), where $s_k^i$ have initially committed before $s_j^i$ initially committed.
MS-SR guarantees the following: (1) the final section of the first transaction, $s_k^f$ , must commit after $s_k^i$. This is the guarantee of multi-stage transactions to commit the initial section before the final section of the transaction. (2) $s_k^f$ must commit before $s_j^f$. This is due to the MS-SR guarantee that the two sections of the transaction must be ordered next to each other relative to other conflicting transactions. (3) $s_k^f$ must be ordered before $s_j^i$ only if there is a conflict between $s_k^f$ and $s_j^i$. This is also due to the need to serialize the sections of two conflicting transactions. The condition of the conflict between $s_k^f$ and $s_j^i$ is to capture that if the two sections do not conflict, then they can be reordered in the serializable history.
These conditions are represented by the following formulation, where (a) captures both conditions (1) and (2), and (b) captures condition (3): 
\begin{align*}
\text{MS-SR: }(a)& \exists t^s\  \left( s_k^i <_{h} s_j^i  \implies ( s_k^i <_h t^s <_h s_j^f \wedge t^s = s_k^f) \right)\\
            (b)&\text{if conflict in } s_k^f, s_j^i \implies s_k^f <_h s_j^i
\end{align*}


We elaborate on Example \ref{ex:multitxn_safety} to demonstrate the need for MS-SR(a) and MS-SR(b).
As an example of MS-SR, consider the two transactions:
\begin{align*}
    \text{$t_k : b^i_{t_k} r^i_{t_k}(x) c^i_{t_k} b^f_{t_k} w^f_{t_k}(x) c^f_{t_k}$ and $t_j : b^i_{t_j} r^i_{t_j}(x) c^i_{t_j}
    b^f_{t_j} w^f_{t_j}(x) c^f_{t_j}$.}
\end{align*}
Further assume that $s^i_k <_h s^i_j.$ Condition MS-SR(a) above guarantees that $s^f_k$ is committed after $s^i_k$ and before $s^f_j$, i.e., we have 
$s^i_k <_h s^f_k <_h s^f_j.$
With MS-SR(a) alone, the following $s^{i}_k <_h s^{i}_j <_h s^{f}_k <_j s^f_j$ is permitted. However, because $s^f_k$ conflicts with  $s^i_j$, then the two sections must be ordered according to MS-SR(b) and the following ordering relations must be met: $s^{i}_k <_h s^{f}_k <_h s^{i}_j <_j s^f_j$. This ordering avoids the anomaly of both transactions reading the same value of $x$, but one overwriting the value written by the other.

Now, we present a protocol that guarantees MS-SR.

\textbf{Two Stage 2PL} (\TSPL): The Two Stage 2PL is the two phase locking protocol~\cite{BernsteinHG87} modified for our multi-stage transactional model (See Algorithm~\ref{algo:TS2PL}.) Let $t_k$ be a multi-stage transaction comprising of $t^i_k$ and $t^f_k$. 
First, the initial section starts executing, locking each accessed data item before reading or writing it. After the initial section finishes processing, the initial commitment cannot be performed immediately. This is because we need to guarantee that the final section can execute and commit as well, due to the requirement of multi-stage transactions. Therefore, the locks of all items that are accessed (or potentially accessed) by the final section must be acquired first. Then, the transaction enters the initial commit phase. Once all the needed input is available for the final section (\emph{e.g.}, the corrected labels from the cloud model), the final section executes, and the transaction enters the final commit phase. Finally, all the locks are released.

\begin{algorithm}[h!]
\SetAlgoLined
 items $\gets$ get\_rwsets($t^i_k$)\\
 \eIf{acquirelocks(items)}{
    execute($t^i_k$)\\
    items $\gets$ get\_rwsets($t^f_k$)\\
 \eIf{acquirelocks(items)}{
    Initial Commit\\
    execute($t^f_k$)\\
    Final Commit
    }{
    abort}
    }{
    abort}
 releaselocks(get\_rwsets($t^i_k$))\\
 releaselocks(get\_rwsets($t^f_k$))\\
 \caption{Two Stage 2PL}
 \label{algo:TS2PL}
\end{algorithm}

\begin{theorem}
The \TSPL\  protocol satisfies MS-SR.
\end{theorem}
\begin{proof}
Consider a pair of conflicting transactions $t_p$ and $t_q$, where $t^i_p <_h t^i_q$.
Following Algorithm \ref{algo:TS2PL}, each section is serialized relative to each other section because locks are held before execution. Now, we show that the three conditions of MS-SR of ordering sections relative to each other are met. The first guarantee is ordering the initial section before the final section. The algorithm executes the initial section before the final section which guarantees their ordering.
The second guarantee that $t_p^f$ is ordered before $t_q^f$. There is at least one data object $o$ that both $t_p$ and $t_q$ access. Because the final section is only executed after all locks are held for the transaction (including the lock for $o$), $t_p^f$ would be processed before $t_q^f$.
The third guarantee is that if $t_p^f$ conflicts with $t_q^i$, then $t_p^f <_h t_q^i$. Assume that the conflict is on data object $o$. Assume to the contrary that $t_q^i <_h t_p^f$. If that's the case, this means that $t_q^i$ acquired the lock on $o$ before $t_p^f$ and before the point of initial commitment (because initial commitment only happens after acquiring all locks including the locks for the final section). Because the locks (including the one on $o$) are not released until $t_q$ finishes, this means that before the lock on $o$ is released, $t_q$ has initially committed. However, $t_p$ initially commits only after acquiring the lock on $o$, which means that $t^i_q <_h t^i_p$, which is a contradiction to our starting assumption that $t^i_p <_h t^i_q$.
\end{proof}

\textbf{Discussion.}
Although MS-SR is an easy-to-use consistency guarantee, it leads to
complications and undesirable performance characteristics. 
The main complication is due to the need to guarantee that committing the initial section would lead to committing the final section. With the stringent requirement that the two sections are serialized so that they appear to be back-to-back in the serialization order, this leads to having to ensure that the locks for the final section can be acquired. The design consequence as we see in the TS-2PL algorithm is that the initial section cannot commit before acquiring the locks of the final section. This leads to one of two consequences: (1)~the system can infer what data will be accessed (or potentially accessed) in the final section so that the locks can be acquired and the initial commit happens before having to wait for the cloud model to finish processing, or (2)~the transaction would not be able to initially commit until the cloud model returns the correct labels so that it is known what data items are going to be accessed. The first option may require complex analysis or input from the programmer and the second option is prohibitive as it means that the initial section has to wait for a potentially long time, which invalidate the goals of multi-stage transactions. 
Another complication is that the locks for the initial section must be held until the final section finishes processing which would lead to higher contention.

%

\subsection{Multi-Stage Invariant Confluence with Apologies (MS-IA)}
Now, we propose a multi-stage safety criterion that is inspired from
invariant confluence~\cite{bailis2014coordination} and apologies~\cite{DBLP:conf/cidr/HellandC09}. The initial-final pattern of multi-stage transactions invites the utilization of these concepts as we discuss next.

\textbf{Guesses and Apologies.}
The concept of guesses and apologies~\cite{DBLP:conf/cidr/HellandC09} was introduced to describe a
pattern of programming that balances local action versus global
action (for example, a local action on a replica versus global action
on the state of all replicas in the context of eventual consistency).
In this pattern, a \emph{guess} is performed with local information
and, then, guesses are reconciled with the global state which would lead
to detecting inconsistencies in the local guesses. Such errors lead to
\emph{apologies} via undoing actions,
administrator intervention, and/or notifications to affected users. 

This pattern of guesses and apologies fits our multi-stage edge-cloud
transaction model. The
initial section represents the guess and the final section represents
the apology. To illustrate, consider an example of a multi-player AR
game with three players: $A$ with 50 tokens, $B$ with 10 tokens, and
$C$ and $D$ with no tokens. The application has a token transfer
function \textsf{transfer(from, to, amount)}. The initial section
performs the transfer, and the final section reconciles any
mistakes. Now, assume that the initial section of a transfer $t_1$ from $A$
to $B$ for 50 tokens took place. Then, the initial section of a
transfer $t_2$ from $B$ to $C$ for 10 tokens took place followed by another
transfer $t_3$ from $B$ to $C$ for 50 tokens. Due to concurrency, assume
that the final section of both $t_2$ and $t_3$ were performed and
that their trigger and inputs were correct. In this case, the final
section terminates for both transactions. Then, the final section of
$t_1$ starts. However, the correct input to $t_1$ turns out to be $D$ instead
of $B$ (for example, because the edge CNN model detected player $B$
when it is actually player $D$ as detected by the cloud CNN model.) 
An apology procedure in the final section could retract
the effects of $t_1$ and any other transactions that depended on it,
which are $t_2$ and $t_3$. 

Using guesses and apologies allows us to process the initial
sections of transactions fast while providing a mechanism to overcome
the mistakes of the edge best-effort computation. However, it may
lead to a cascade of retractions. To overcome this, we propose
combining the concept of apologies with invariant confluence as we
show next.

\textbf{Invariant Confluence.}
In invariant confluence, preserving the application-level invariants is
what constitutes a safe execution. In its original form, invariant
confluence is intended to reason about transactions mutating the
state of different copies of data~\cite{bailis2014coordination}. Our edge-cloud model is
different, involving mutating the state of one (edge) copy. However,
an inconsistency might be introduced by the initial section of a
transaction with erroneous trigger/input. Our insight is that we can
utilize the final (apology) section to act as the \emph{merge} function that
attempts to reconcile application-level invariants instead of all
potential inconsistencies. In a way, we are flipping the model of
invariant confluence systems from a pattern of \emph{check-then-apply}
(check if the operation can merge, and decide whether coordination is
necessary before doing the operation), to a pattern of
\emph{apply-then-check} (do the operation then check whether you can merge, and if you cannot merge, then perform an apology procedure and retract the initial section's effects.) 

\textbf{MS-IA programming pattern.}
This pattern, when combined with apologies, can lead to reducing the
negative consequences of erroneous triggers/inputs. Consider the
multi-player AR game application introduced above (when discussing
apologies). Assume that
the initial sections of $t_1$,
$t_2$, and $t_3$, were processed as well as the final sections of
$t_2$ and $t_3$. At this stage, $A$, $B$, and $D$ have no tokens and
$C$ has 60 tokens. When the error is discovered, it triggers the
final section of $t_1$. A programmer, equipped with the notions of
invariant confluence and apologies, writes the final section to attempt
to perform two tasks: (1)~retract the minimum amount of erroneous
actions and their effects using apologies, and (2)~retain as much
state as possible using invariant-preserving merge functions. The
specifics of this pattern depends on the application invariants. For
example, the final section of the transfer tasks could have the 
invariant that no player should
have less than 0 tokens. The final section of $t_1$ would retract the
50 tokens that were initially sent to $B$ and sends them to the
rightful recipient, player $D$. This means that $B$ could not have
sent a combined 60 tokens to $C$. The merge function can then decide
to retain the 10 tokens sent from $B$ to $C$, since they are not
affected by the error. But, it retracts the 50 tokens. This
retraction is accompanied by an apology that depends on the
application (\emph{e.g.}, a message is sent to both $B$ and $C$, with
a free game item.) 

In terms of the concurrency control guarantee that is needed for MS-IA, the initial section of a transaction must be ordered before its corresponding final section (in addition to our earlier assumption that each section is serialized relative to other transactions' sections). Formally, for an initial section, $s_k^i$, the following is true:
 \begin{center}
         MS-IA: $ \exists t^s\ \left( s_k^i <_h t^s \wedge t^s = s_k^f \right)$
     \end{center}


\begin{algorithm}[]
\SetAlgoLined
 items $\gets$ get\_rwsets($t^i_k$);\\
 \If{acquirelocks(items)}{
    execute($t^i_k$)}
 Initial Commit\\
 releaselocks(get\_rwsets($t^i_k$))\\

 items $\gets$ get\_rwsets($t^f_k$)\\
 \eIf{acquirelocks(items)}{
    execute($t^f_k$)}{
    abort}
 Final Commit\\
 releaselocks(get\_rwsets($t^f_k$))
 \caption{MS-IA Algorithm}
 \label{alg:msia}
\end{algorithm}

\textbf{Concurrency control.}
The concurrency control algorithm starts by acquiring all the locks for the initial section, then processing the initial section. When the processing of the initial section is done, the locks are released. Then, when the final section is ready to start, the corresponding locks are acquired before processing the final section. Finally, the locks for the final section are released. Note here that unlike the algorithm for MS-SR, we did not hold the locks for the initial section until the end of the final section and we reach the point of initial-commit immediately after processing the initial section without having to wait to lock or coordinate the final section. The reason for this is that the logic for invariant checking and apologies is embedded in the final section and that we do not need to ensure that the initial and final sections of one transaction are serialized next to each other.

\textbf{Discussion.}
To have better performance characteristics, MS-IA presents a more complex programming abstraction than MS-SR because it places the burden of coordination (invariance checking, reconciliation, and apologies) on the programmer.
In MS-IA, transactions are written as guesses (in the initial section) and apologies (in the final section). Furthermore, apologies are merge functions that aim to reconcile the inconsistencies caused by incorrect triggers or inputs. Given our apply-then-check pattern, it is possible that some operations cannot be merged. In such cases the final section would undo the effects of the initial section---and any transactions dependent on it. 
We envision that this pattern of multi-stage guesses and apologies can incorporate advances in merge operators that would allow minimizing the need for undoing transactions. For example, programmers may use merge-able operations in the initial sections and delaying other operations to the final section. This can benefit from---and help empower---the literature of conflict-free and compositional data types. These can be adapted to the initial-final pattern by making merge-able parts in the initial section and enabling other types of operations in the final section.

 \color{black} In Validation-based (optimistic) protocols, which operate in the context of a single transaction, before validation, the outcome of the transaction is not returned to the client and is not exposed to other transactions. Applying validation-based protocol as they are in the edge-cloud setting would be prohibitive because it means that a transaction would not commit until the validation step - that would happen after cloud processing - is ready. The MS-IA pattern, on the other hand, divides the transaction logic to two sections each acting as an independent transaction, where the first one commits before the second section starts, which allows returning responses to clients and exposing the outcome to other transactions (even before the final section and without having to wait for the processing at the cloud).\color{black}

{\color{black}
\subsection{Multi-Partition Operations}
The transaction processing protocols presented in this section focus on transactions that are local to a partition. In the case of distributed transactions (spanning multiple partitions), the presented algorithms need to be extended. In particular, in the multi-partition case, the data objects that are accessed by a transaction (whether in the initial or final sections) can be in multiple partitions. Locking data objects in remote partitions will be performed by sending the lock requests to the remote edge node that is responsible for the partition. The second difference is that after the transaction finishes, the partitions engage in a two-phase commit protocol to ensure that the distributed commit is performed in an atomic way. This atomic commitment step is performed in the following cases: (1)~for MS-SR, it is performed at the end of the final section, (2)~for MS-IA, it is performed at the end of both the initial and final sections. The reason for not performing this step at the end of the initial section in MS-SR is that the locks are not released until the end of the corresponding final section.
}

\begin{figure*}
\centering
\includegraphics[scale=0.45]{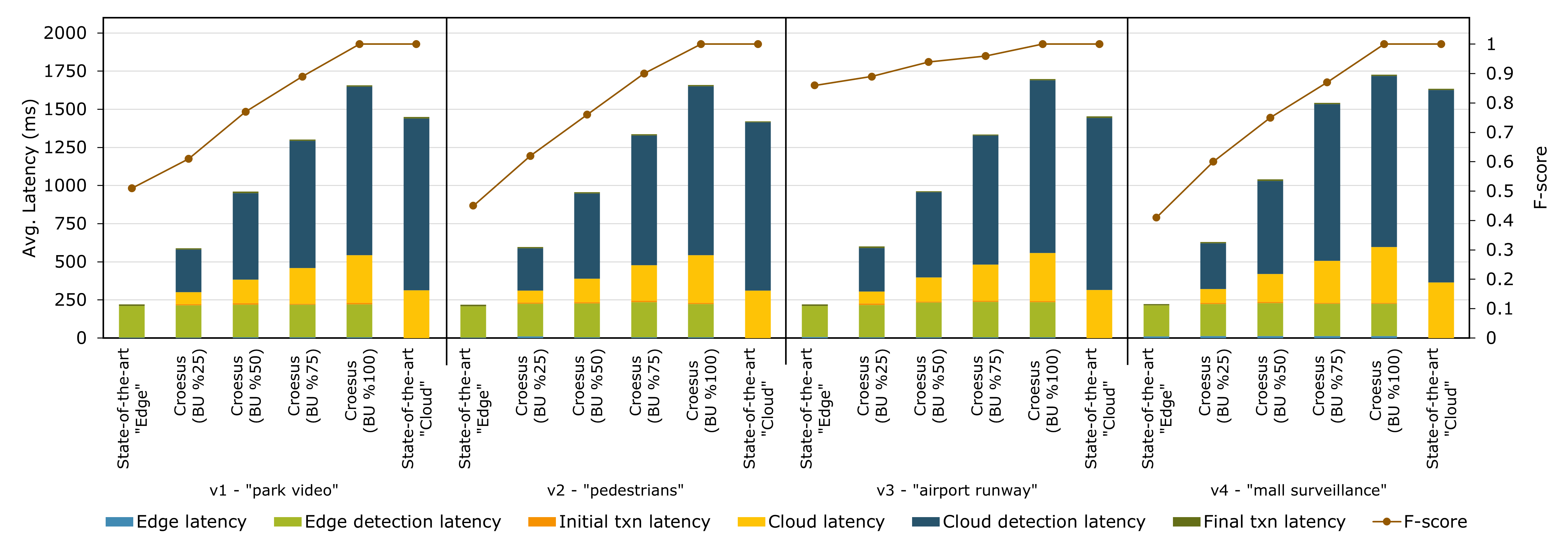}
\caption{Croesus vs. state of the art baselines: Latency and F-score of running Croesus over four videos. Some values are minute and are hard to show on the figure.}
\label{fig:CroesusBaseline}
\end{figure*}

\section{Evaluation}
\label{sec:eval}


In this section, we show how Croesus manages the trade-off between performance and accuracy of two models with different characteristics: (1)~YOLOv3~\cite{yolo9000,yolov3} as the cloud model, which is reported to achieve 45~FPS on high-end hardware and achieves high accuracy. (2)~Tiny YOLOv3~\cite{yolov3,yolo9000}---which is a compact version of YOLOv3---for the edge model. Tiny YOLOv3 is faster but less accurate than YOLOv3~\cite{yolov3}.
%

We compare Croesus with two baselines: 
        •	State-of-the-art edge baseline: this baseline represents a performance-centric video analytics applications where a compact model (Tiny YOLOv3) is deployed on the edge machine for lower latency.
        •	State-of-the-art cloud baseline: this baseline represents accuracy-centric video analytics applications where a computationally expensive model (YOLOv3) is deployed on a resourceful cloud machine for better accuracy.

\begin{figure}[!t]
\centering
\includegraphics[scale=0.45]{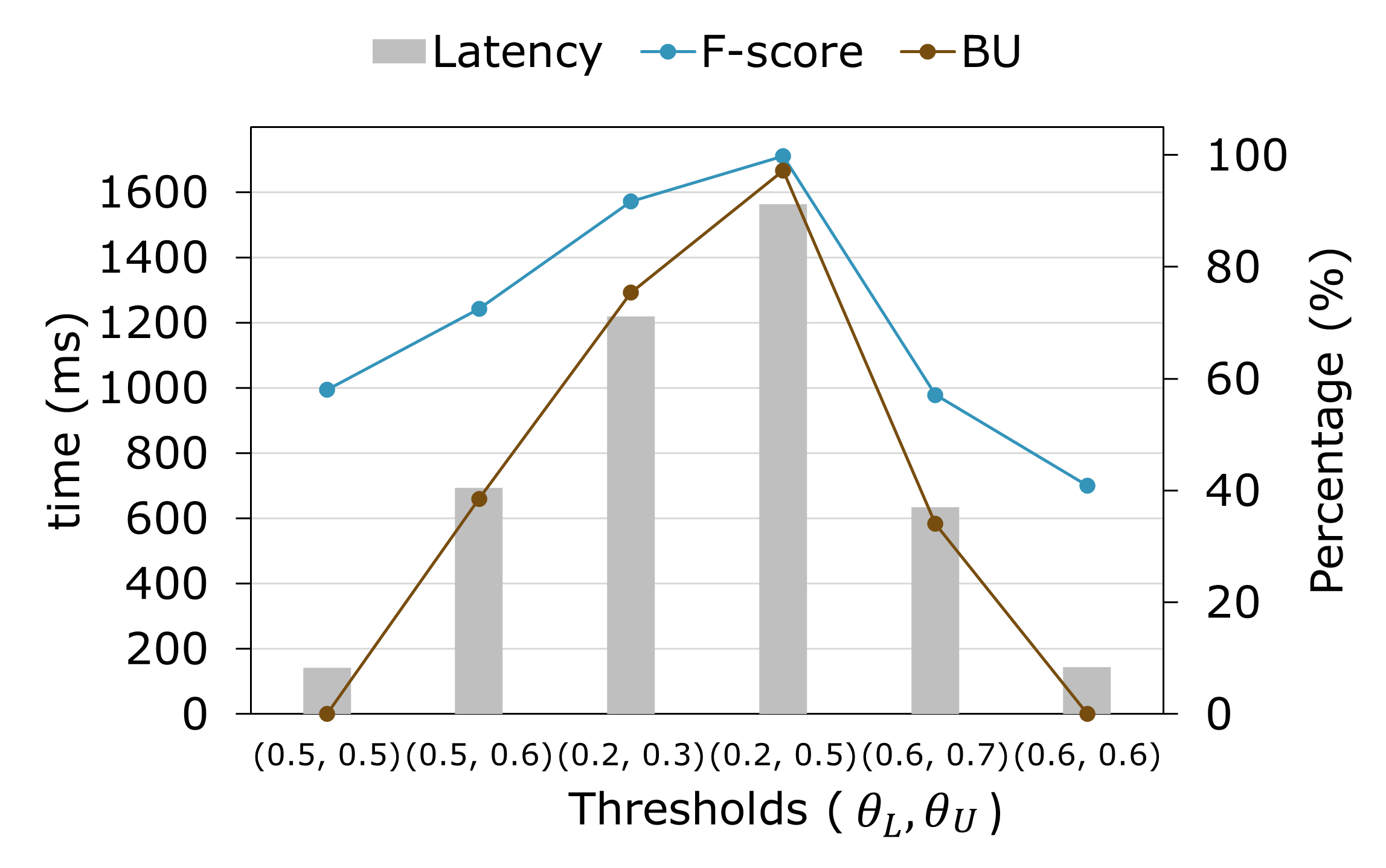}

\caption{Croesus latency vs. accuracy for different pairs of thresholds 
}
\label{fig:CroesusThresh}
\vspace{-0.15in}
\end{figure}

\begin{figure*}
\includegraphics[scale=.63]{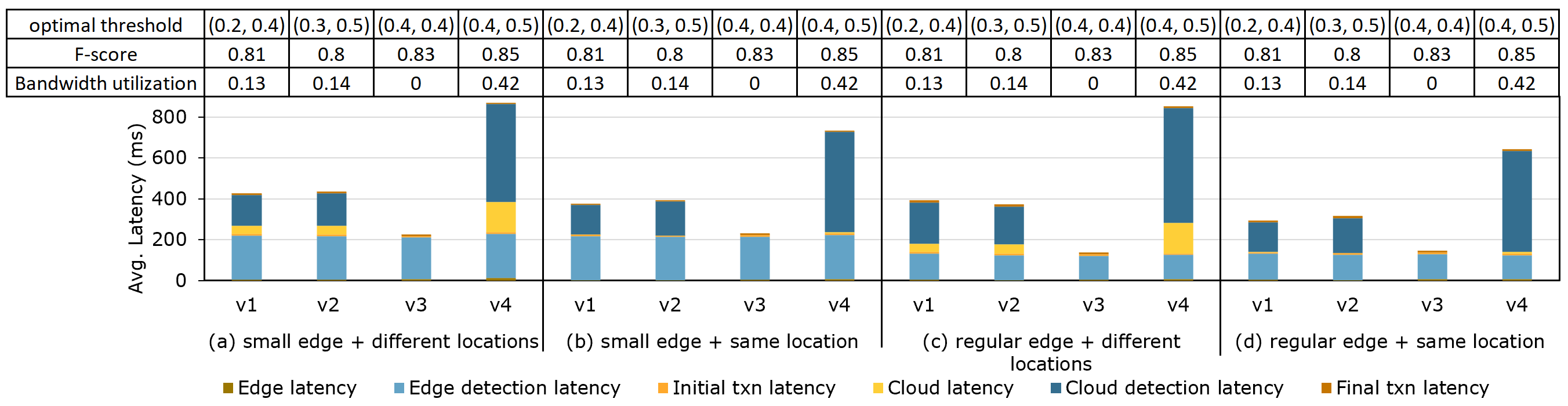}

\caption{Latency in different setups for the optimal case that was dynamically configured by Croesus.}
\label{fig:CroesusSetups}
\end{figure*}

\begin{table}[]
\caption{Comparison between state-of-the-art edge and cloud and optimal threshold Croesus}
\label{tab:my-table}
\begin{center}
\begin{tabular}{c|c|c|c|c|c|c|}
\cline{2-7}
                         & \multicolumn{3}{c|}{Accuracy} & \multicolumn{3}{c|}{Latency (ms)}                                            \\ \cline{2-7} 
                         & Croesus   & Edge    & Cloud   & Croesus                                                   & Edge   & Cloud   \\ \hline
\multicolumn{1}{|c|}{v1} & 0.81x     & 0.5x    & 1       & \begin{tabular}[c]{@{}c@{}}427.02\\ (226.16)\end{tabular} & 210.74 & 1452.5  \\ \hline
\multicolumn{1}{|c|}{v2} & 0.8x      & 0.45x   & 1       & \begin{tabular}[c]{@{}c@{}}434.81\\ (224.41)\end{tabular} & 207.97 & 1427.69 \\ \hline
\multicolumn{1}{|c|}{v3} & 0.83x     & 0.86x   & 1       & \begin{tabular}[c]{@{}c@{}}225.63\\ (218.17)\end{tabular} & 211.19 & 1455.66 \\ \hline
\multicolumn{1}{|c|}{v4} & 0.85x     & 0.41x   & 1       & \begin{tabular}[c]{@{}c@{}}863.96\\ (235.02)\end{tabular} & 214.65 & 1638.89 \\ \hline
\end{tabular}
\end{center}
\end{table}

\cut{
\begin{table}[h!]

\begin{center}
\fontsize{8pt}
\resizebox{\columnwidth}{}{%
\begin{tabular}{l|c|c|c|c|c|c|}
\cline{2-7}
                         & \multicolumn{3}{c|}{\textbf{Accuracy}} & \multicolumn{3}{c|}{\textbf{Latency (ms)}}                                            \\ \cline{2-7} 
                         & \textbf{\textit{Croesus}}   & \textbf{\textit{Edge}}    & \textbf{\textit{Cloud}}  & \textbf{\textit{Croesus}}                                              & \textbf{\textit{Edge}}   & \textbf{\textit{Cloud}}   \\ \hline
\multicolumn{1}{|l|}{v1} & 0.81x     & 0.5x    & 1      & \begin{tabular}[c]{@{}c@{}}427.02 \\ (226.16)\end{tabular} & 210.74 & 1452.5 \\ \hline
\multicolumn{1}{|l|}{v2} & 0.8x      & 0.45x   & 1      & \begin{tabular}[c]{@{}c@{}}434.81\\ (224.41)\end{tabular}  & 207.97 & 1427.69 \\ \hline
\multicolumn{1}{|l|}{v3} & 0.83x & 0.86x & 1 & \begin{tabular}[c]{@{}c@{}}225.63\\ (218.17)\end{tabular} & 211.19 & 1455.66 \\ \hline
\multicolumn{1}{|l|}{v4} & 0.85x     & 0.41x   & 1      & \begin{tabular}[c]{@{}c@{}}863.96\\ (235.02)\end{tabular}  & 214.65 & 1638.89  \\ \hline
\end{tabular}%
}
\end{center}
\end{table}
}

\subsection{Experimental setup}

Our evaluations are performed on Amazon's AWS EC2 services. 
Edge machines are implemented on either t3a.xlarge instances (for the default setups) and t3a.small (for experiments with limited resources). t3a.small machines have 2 virtual CPUs and 2GiB of memory and t3a.xlarge machines have 4 virtual CPUs and 16GiB of memory. Machine locations are either in California or Virginia. The default setup is of an edge machine in California and a cloud machine in Virginia.
We implement a prototype of Croesus in Python. 
In addition to model detection, the edge node maintains a data store and {\color{black}processes transactions according to the MS-IA algorithm.} Transactions are constructed by randomly selecting keys to read or write to the database in response to detected labels. 



We evaluate accuracy and performance as follows: Accuracy is measured as the F-score. Performance is measured in two ways: (1) Latency, which we define as the time required to commit transactions in the system. (2) Edge-Cloud Bandwidth Utilization (BU), which we define as the ratio of frames being sent to the cloud relative to all processed frames. {\color{black}This metric is proportional to the number of corrections that need to be made in the final transaction.} We consider the YOLOv3 output to be the ground truth and we use it to compare Creosus' results and calculate the F-score. When the overlap between the truth boundaries and the predicted boundaries is more than \%10, we consider the prediction correct. {\color{black}The calculation of the F-Score does not depend on the percentage of frames that are sent or not sent to the cloud, but rather on the accuracy of the detection from the perspective of the client (i.e., the accuracy of the detection \emph{and} apologies, if any.) There is, however, a correlation between sending more frames to the cloud as it means that more errors are corrected by the more accurate cloud model.}
 

Experiments run on a subset of five types of videos: Street traffic (vehicles), street traffic (pedestrians), mall surveillance (all three querying for 'person'), airport runway querying for 'airplane', and home video of pet in the park querying for 'dog'. Each detection acquired for each frame triggers a transaction that has 6 operations, half of these mutate the state of the database by inserting data items, and the other half read from previously added items. This mimics a write-heavy workload of YCSB (Workload A)~\cite{Cooper2010}. {\color{black}Unless we mention otherwise, we use MS-IA as the consistency guarantee.}

\vspace{-0.05in}

\subsection{Experimental results}
\subsubsection{\textbf{Performance vs. accuracy trade-off}}
Figure \ref{fig:CroesusBaseline} shows the trade-off between the latency and accuracy as BU varies on four videos: park video (v1), street traffic (v2), airport runway (v3) and mall surveillance (v4). For each video, we compare different BU configurations with the state-of-the-art edge and cloud solutions. In the figure, the stacked bars represent the latency breakdown for each experiment. Edge latency and cloud latency represent the average time needed to send a frame to the edge and to the cloud, respectively. The edge detection latency and cloud detection latency are defined as the average time it takes the tiny YOLOv3 and YOLOv3 models, respectively, to produce the detected objects list in a frame. The initial transaction and final transaction latency are very minute and hard to show in the figure, but they represent the time it takes to commit a transaction after detection is done. The F-score metric is shown as a marked line. 

{\color{black} As shown in Figure~\ref{fig:CroesusBaseline}, Croesus processes transaction updates in the initial phase (measured by edge latency and edge detection latency), up to $6.9\times$ faster than the case with full BU while maintaining high accuracy (F-score up to $\%94$ in the case of "airport runway") by utilizing the cloud corrections and final transaction. } The client observes two latencies: the first is the real-time initial processing at the edge which corresponds to edge latency, edge detection latency, and initial transactions latency. The second is for the final processing after corrections, if any, from the cloud, which corresponds to all the latency types shown in the figure. As BU increases, the amount of frames sent to the cloud, and consequently the average cloud-related latencies, increases. When BU is 100\%, the total cloud latency for Croesus becomes even higher than state-of-the-art cloud because it incurs all the overheads of the state-of-the-art cloud in addition to the overhead of Croesus methods. 


The trend of increasing Croesus cloud latency as BU increases is observed in videos 1, 2, and 4. However, a unique trend appears for video 3 (querying for ‘airplane’ on the airport runway video). In this video, the state-of-the-art edge produce high accuracy due to the nature of the video (an object that is detected by the edge model with high confidence). This asserts the need for dynamic optimization over the detection thresholds for different applications in order to address workload differences. Croesus’ dynamic optimization ensures the best balance of the trade-off between accuracy and latency depending on the needs of each application. 


Figure \ref{fig:CroesusThresh} demonstrates the effect of choosing different thresholds on the latency in Croesus.  {\color{black}We demonstrate the results using the street traffic video querying for vehicles.} It shows the total Croesus cloud latency and the BU percentage as the threshold pairs for detections are varied. For example, a threshold pair (0.5, 0.6) means that only detections with confidence values in the edge mode that are within these two values are sent to the cloud for verification. Detections with lower confidence values are discarded and ones with higher confidence values are assumed correct by the edge node and are not verified (however, erroneous detections are still accounted for in the F-score.)

When the thresholds are set to (0.5, 0.5) the resulting BU is $\%0$ since no frames will be sent to the cloud for validation. The resulting accuracy is comparable to the edge only baseline at $\%58$. For a threshold pair of (0.5, 0.6), the latency increases due to more results being validated in the cloud. The resulting BU is $\%38.5$ while the F-score increases by $\%25$. When the BU reaches $\%97.2$, the accuracy reaches $\%99.8$. 
For thresholds (0.6,0.7), the BU is only $\%4$ lower than the BU of the thresholds (0.5, 0.6). However, the F-score decreases by more than $\%21.24$. This shows that although two pairs may have similar BU values, their corresponding F-score can be significantly different. It indicates the importance of dynamically optimizing for an optimal pair of thresholds that balance the trade-off between the latency and accuracy while prioritizing thresholds that yield higher accuracy.


\begin{table}[]
\caption{The effect of the cloud model size.}
\centering
\label{tab:model-size}
\begin{tabular}{|l|l|l|l|l|}
\hline
\textbf{\begin{tabular}[c]{@{}l@{}}Croesus\\ cloud model\end{tabular}} & \textbf{\begin{tabular}[c]{@{}l@{}}Optimal\\ threshold\end{tabular}} & \textbf{\begin{tabular}[c]{@{}l@{}}Croesus\\ F-score\end{tabular}} & \textbf{\begin{tabular}[c]{@{}l@{}}Bandwidth\\ Utilization\end{tabular}} & \textbf{\begin{tabular}[c]{@{}l@{}}Detection\\ latency (sec)\end{tabular}} \\ \hline
YOLOv3-320                                                             & (0.2, 0.3)                                                           & 0.84                                                               & 0.61                                                                     & 0.70                                                                       \\ \hline
YOLOv3-416                                                             & (0.4, 0.5)                                                           & 0.86                                                               & 0.44                                                                     & 1.12                                                                       \\ \hline
YOLOv3-608                                                             & (0.4, 0.6)                                                           & 0.83                                                               & 0.58                                                                     & 2.34                                                                       \\ \hline
\end{tabular}
\end{table}

Another observation from Figure~\ref{fig:CroesusThresh} is that the rate at which the bandwidth utilization increases is faster than the rate of F-score increase over different threshold pairs. This is an indicator that increasing dependence on the cloud does not necessarily improve accuracy dramatically. 

\begin{figure*}
\centering
\includegraphics[scale=0.7]{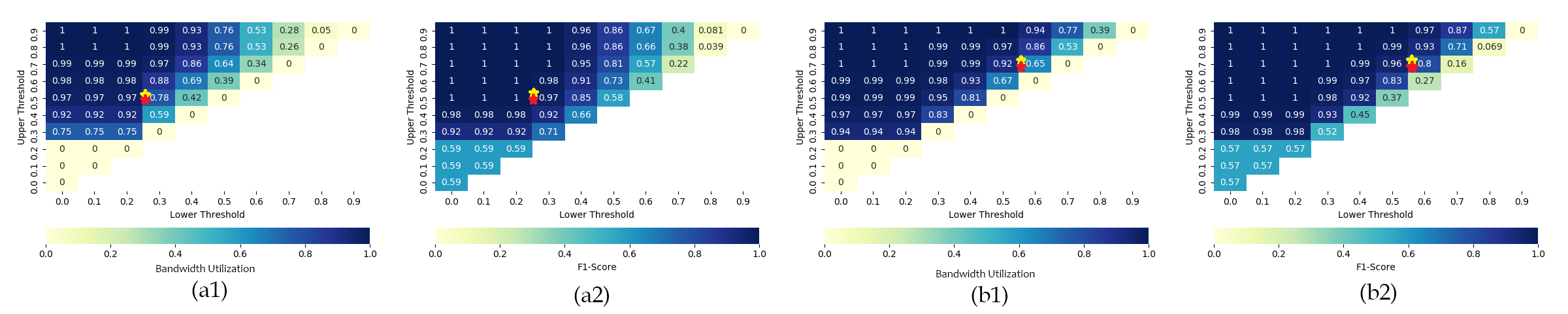}

\caption{Croesus bandwidth utilization vs. accuracy based on the threshold pair choice. a) traffic video querying "person" ($\mu = 0.90$) and b) mall surveillance querying "person" ($\mu = 0.80$). For all pairs of lower threshold $(\theta_L)$ and upper threshold $(\theta_U)$. Dynamically chosen
pair: yellow star using brute force, red star using gradient step.}
\label{fig:heatmaps}
\end{figure*}

{\color{black} The effect of changing the cloud model size in Croesus is demonstrated in Table~\ref{tab:model-size}. 
In this experiment, we set $\mu = 0.8$ and compare the performance of Croesus while using three different cloud model sizes: YOLOv3-320, YOLOv3-416, YOLOv3-608, where the number at the end of each model's name represents the width and height used in the neural network model. Therefore, a larger number indicates a larger model. 
As the cloud model size gets larger, the detection latency gets larger as well. This is the main impact of utilizing different model sizes. The different models have different accuracy characteristics as well. However, using them in the Croesus framework does not demonstrate such differences in the resulting F-score and BU. This is because the optimal thresholds are set based on the used cloud model to achieve the desired minimum accuracy, $\mu$. }

\subsubsection{\textbf{Optimal threshold performance on different setups}}
Figure \ref{fig:CroesusSetups} shows the accuracy and performance results of Croesus for different videos when using the optimal threshold. These experiments run across four different setups: (a) Small edge, different locations: Edge machines are of type t3a.small while cloud machines are of type t3a.xlarge. Edge machine are located in California and cloud machines are in Virginia. (b) small edge, same location: Small edge, different locations: Edge machines are of type t3a.small while cloud machines are of type t3a.xlarge.Edge and cloud machines are physically located in the same location. (c) Regular edge, different location: Edge and cloud machines are both of type t3a.xlarge. Edge machine are located in California and cloud machines are in Virginia. (d) Regular edge, same location: Edge and cloud machines are physically located in the same location and are both of type t3a.xlarge.

This figure demonstrates the improvement in latency that the optimal thresholds provide compared with the performance shown in Figure \ref{fig:CroesusBaseline} {\color{black}(For a clearer presentation, we show the comparison numbers in Table~\ref{tab:my-table}, where the number inside the parentheses in Croesus is the latency of the initial transaction.)}. Also, it shows the effect of resource allocation and geographical location on performance, and the importance of dynamic threshold optimization to address the differences in applications. 

In the case of applying the optimal thresholds, we see improvement in {\color{black}the final} latency over the state-of-the-art cloud implementation by up to $\%85$ (but as low as $\%47$ for the case of v4). In addition, committing the initial transaction is always comparable to the state-of-the-art edge solutions. Even though the final transaction in Croesus can take up to $\%75$ more than the edge only implementation, the accuracy improvements is significant and can justify the slight delay after the initial transaction. 

In addition, the F-score of optimal Croesus is 2.1x higher than the F-score of edge-only in video v4. In the case of video v3, the accuracy is comparable to the state-of-the-art accuracy because the optimal thresholds represent a near $\%0$ bandwidth utilization. This is possible in application where objects are expected to be easier to detect in each frame. 
The figure also shows that as the geographical distance between the edge and the cloud decreases (when placed in the same location), Croesus performance improves. In addition, the performance improves when edge resources are maximized.

\begin{figure*}
\centering
\includegraphics[scale=0.4]{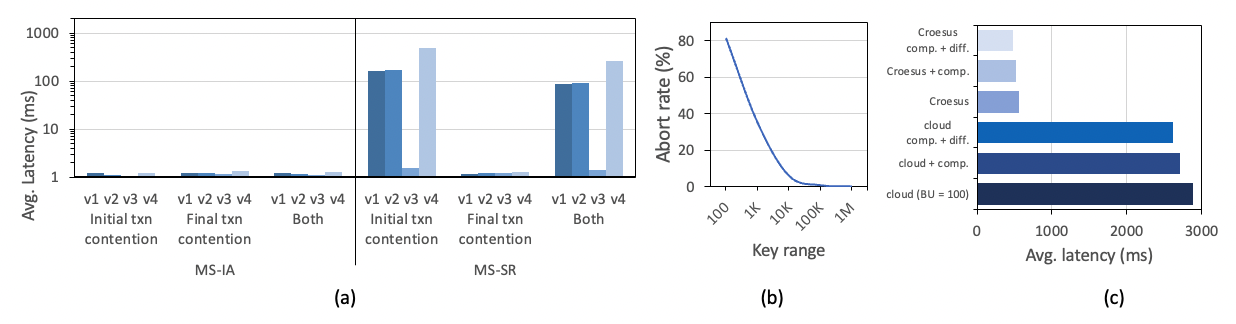}
\caption{\color{black}(a) Comparing lock contention of MS-SR and MS-IA measured as the average latency of holding locks. (b) Abort rate of MS-SR transactions. (c) Hybrid system techniques.}
\label{fig:contention + compress}
\end{figure*}

\subsubsection{\textbf{Dynamic preprocessing optimization}}
Figure \ref{fig:heatmaps} shows the bandwidth utilization and accuracy as we vary the optimization thresholds (the lower threshold $\theta_L$ and upper threshold $\theta_U$). The heatmaps illustrate the gradual shift in the balance between bandwidth utilization and accuracy.

\textbf{bandwidth utilization/accuracy trade-off.}
Figure \ref{fig:heatmaps}(a1) for BU and Figure \ref{fig:heatmaps}(a2) for F1-Score show the trend where increasing the lower threshold and the gap between the two optimization thresholds results in a higher throughput. For example, when the optimization pair is (0.2, 0.4), the F-score is $\%98$ since this pair of thresholds result in a high BU at$\%92$. However, when the optimization pair is (0.3, 0.4) the bandwidth utilization drops to $\%59$ while the F-score remains relatively high $\%92$. We are able to conserve the edge-cloud communication by more than \%35.9 while maintaining relatively high accuracy.

Figures \ref{fig:heatmaps}(b1) for BU and \ref{fig:heatmaps}(b2) for F1-Score show the same trends as the previous set of heatmaps. However, we notice a sudden jump in bandwidth utilization and F-score results. This is due to the quality of this second video where objects are smaller and not as clear as the first video. In this case, utilizing edge-cloud communication increases the quality of detections dramatically compared to edge detections. For example, for the optimization pair (0.4,0.5) \%81 of frames are sent to the cloud and the F-score is \%92. However, when the optimization pair is (0.4,0.4) no frames are sent to the cloud and the F-score decreases to \%45. 

{\color{black} \textbf{Dynamically finding the optimal solution.} We implemented two approaches to acquire the optimized pair of thresholds. The first is a brute force method that evaluates the whole space of threshold pairs. In it, we obtain the optimal pair for balancing the trade-off (shown as a yellow star). The second approach uses a gradient step with our optimization formulation. Using gradient step is 2.2x times faster (shown as a red star). In both cases, bandwidth utilization is $<\%78$, accuracy is at least $\%49$ higher than an edge model.}

\subsubsection{\textbf{\color{black}Comparing MS-SR and MS-IA}}
{\color{black}
In the next set of experiments, we measure the performance differences between the two proposed consistency levels: MS-SR and MS-IA. (In this set of experiments we use video v4 with the query ``person''.) The main difference between the two consistency levels is that the locks in the initial section of MS-SR are held until the end of the whole transaction, whereas in MS-IA, the locks are released after the initial section. This results in increasing the lock contention in MS-SR. Figure~\ref{fig:contention + compress}(a) shows the difference in contention by measuring the average time locks are held in MS-SR and MS-IA (denoted average latency in the figure.) While the average latency of MS-IA is in the order of milliseconds, the average latency of initial sections in MS-SR is in the order of hundreds of milliseconds. This is because the locks are not released until the final section is performed which means that the locks are held while the frame is being processed using the cloud model which takes a significant amount of time.

The contention difference leads to a high likelihood of aborts in MS-SR. Figure~\ref{fig:contention + compress}(b) shows the abort rate of transactions in MS-SR while emulating a high contention scenario of hot sports with different sizes. The x-axis (key range) is the key range of the hot spot that the transactions are trying to access. In this model, transactions are executed in batches of 50 transactions per batch where each transaction has 5 update operations. The figure shows that the abort rate can be significant when the hot spot has a size that is less than 10K keys. This demonstrates the benefit of using MS-IA to overcome the hot spot contention problems while using MS-SR.
The figure does not show the abort rate of MS-IA transactions as the rate is 0\% for all cases. This is because our implementation uses a single-threaded sequencer to order transactions in batches so that conflicting transactions do not overlap. This is possible as the transactions do not have to hold locks for prolonged durations.
}
\subsubsection{{\color{black}Hybrid edge-cloud techniques}}
\label{subsub:hybrid}
{\color{black}
Hybrid edge-cloud techniques have been proposed to process object detection models~\cite{glimpse, collaborativeedge, noscope, neurosurgeon}. These techniques generally work by performing some pre-processing steps at the edge node before sending the frame to be detected at the cloud. We compare with two such techniques that were utilized in various forms in prior work~\cite{collaborativeedge,noscope}: (1)~\emph{compression} in which the frame is compressed before sending it to reduce the communication bandwidth and latency, and (2)~\emph{difference communication} in which only the difference between the current frame and a reference frame is sent to the cloud. These techniques, if implemented in isolation, would achieve a small improvement over the performance of the state-of-the-art cloud baseline that we compared with as they would still require sending all frames for detection in the cloud. We show this in the evaluations on the park video v1 with the larger cloud model (YOLOv3-608) in Figure~\ref{fig:contention + compress}(c) under cloud+compression and cloud+compression+difference. These evaluations apply the hybrid techniques which improves the latency as less data need to be sent. However, this is a small improvement because the latency is dominated by the detection latency at the cloud. 

An alternative view of these techniques is as methods to augment with edge-cloud Croesus. Figure~\ref{fig:contention + compress}(c) also shows how augmenting compression can improve the final commit latency in Croesus (under Croesus+compression and Croesus+compression+difference). The improvement is small because the model detection latency in the cloud is the dominant latency (as we show in previous evaluations.)
}

\section{Related Work}
\label{sec:related}

%
The requirement of real-time processing has been tackled by real-time Databases (RTDB)~\cite{son1996improving} that aim to process data in predictable short time. Our method differs by allowing to manage the trade-off of performance and accuracy and providing the illusion of both a fast and accurate processing.
A hybrid edge-cloud model (and similar caching-based models) have recently been used~\cite{glimpse, collaborativeedge, noscope, neurosurgeon} to take advantage of cloud computing to process data on neural networks, as well as leveraging resources at the edge. Our work extends these efforts by providing a multi-stage transactional model that enables programmers to reason about this hybrid edge-cloud model. {\color{black} In particular, these hybrid edge-cloud models can be augmented with the edge-cloud model of Croesus to improve the edge-to-cloud latency. However, when hybrid edge-cloud models are used in isolation, they would incur the high costs of edge-to-cloud communication for all frames since they require performing the detection in the cloud.}

The multi-stage transaction model differs from existing abstractions in that each transaction is split into two asymmetrical sections. This makes traditional consistency models~\cite{bernstein1987concurrency} unsuitable for multi-stage transactions. The pattern of initial-final sections resemble work on eventual consistency~\cite{bailis2013eventual} and Transaction Chains~\cite{zhang2013transaction} but differs in one main way: the inconsistencies in the multi-stage model are external to the database. They are caused by erroneous inputs or triggers. In eventual consistency and Transaction Chains, inconsistency is caused by concurrent operation across different copies. This leads to similarities and differences, which led us to adapt prior relevant literature. 
Multi-stage transactions resemble work on long-lived transactions (LLT) as well, such as Sagas~\cite{sagas}. Multi-stage transactions can be viewed as a special case of LLT's---with a transaction and a follow-up correction/compensation transaction---which enables simpler and more efficient solutions. 

\color{black}We view Croesus as a data layer solution that builds on top of asymmetric environments which - like edge-cloud - may include the lambda architecture~\cite{7364082} with both batch processing (slower but more accurate) and speed/real-time processing (faster but less-accurate). The contributions of Croesus can be applied to the lambda environment~\cite{warren2015big} by using multi-stage transactions (where the initial section is processed after real-time processing and the final section is processed after batch processing), and thus provide Croesus benefits to lambda programmers.\color{black}

\section{Conclusion}
\label{sec:conclusion}
    We presented Croesus, a multi-stage processing system for video analytics and a multi-stage transaction model which optimizes the trade-off between performance and accuracy. We present two variants of transnational consistency for multi-stage transactions---multi-stage serializability and multi-stage invariant confluence with apologies. Our evaluation demonstrates that multi-stage processing is capable of managing the accuracy-performance trade-off and that this model provides both immediate real-time responses and high accuracy.
    
    Although we have presented the concept of multi-stage processing and transactions in the context of edge-cloud video analytics and processing~\cite{nawab2021wedgechain,nawab2018dpaxos,nawab2018nomadic,gazzaz2019collaborative,mittal2021coolsm}, these concepts are relevant to many problems that share the pattern of needing immediate response and complex processing. Our future work explores these applications. One area of future work is to apply this pattern of multi-stage processing to blockchain systems with off-chain components~\cite{abadi2020anylog,alaslani2019blockchain,nawab2019blockplane}. In such a case, the first stage is performed in the off-chain component while the final stage is performed after validation from the blockchain. Another area we plan to explore is to integrate the multi-stage processing structure with global-scale edge placement and reconfiguration~\cite{zakhary2018global,zakhary2016db}. This will allow utilizing multi-stage processing more efficiently by controlling where the stages are performed and what edge/cloud datacenters to utilize.
    
\section{Acknowledgement}
This research is supported in part by the NSF under grant CNS-1815212.
\bibliographystyle{IEEEtran}
\bibliography{references}

\end{document}